\documentclass[aps,pra,groupedaddress,showpacs,superscriptaddress,nofootinbib,twocolumn]{revtex4-1}
\usepackage{lineno}
\usepackage{amsmath}
\usepackage{amssymb}
\usepackage{amsthm}
\usepackage[hyperindex,breaklinks]{hyperref} 

\newcommand{\ZZ}{\mathbb{Z}}

\newcommand{\CC}{\mathcal{C}}

\newcommand{\HC}{\mathcal{H}}

\newcommand{\OC}{\mathcal{O}}
\newcommand{\PC}{\mathcal{P}}

\newcommand{\SC}{\mathcal{S}}

\newcommand{\WC}{\mathcal{W}}

\newcommand{\ket}[1]{|#1\rangle}                  
\newcommand{\bra}[1]{\left\langle #1 \right|}     
\newcommand{\dyad}[2]{\ket{#1}\bra{#2}}           
\newcommand{\Tr}{{\rm Tr}}                        
\newcommand{\vect}[1]{\mbox{\textbf{#1}}}         

\newcommand{\ii}{\mathrm{i}}					  

\long\def\ca#1\cb{} 

\newtheorem{theorem}{Theorem}

\begin{document}

\title{Standard form of qudit stabilizer groups}

\author{Vlad Gheorghiu}
\email{vgheorgh@andrew.cmu.edu}
\affiliation{Department of Physics, Carnegie Mellon University,
Pittsburgh,
Pennsylvania 15213, USA}

\date{Version of \today}

\begin{abstract}
We investigate stabilizer codes with carrier qudits of equal
dimension $D$, an arbitrary integer greater than 1. We prove
that there is a direct relation between the dimension of a qudit
stabilizer code and the size of its corresponding stabilizer,
and this implies that the code and its stabilizer are dual to
each other. We also show that any qudit stabilizer can be put in
a standard, or canonical, form using a series of Clifford gates,
and we provide an explicit efficient algorithm for doing this.
Our work generalizes known results that were valid only for
prime dimensional systems and may be useful in constructing
efficient encoding/decoding quantum circuits for qudit
stabilizer codes and better qudit quantum error correcting
codes.
\end{abstract}

\maketitle

\section{Introduction\label{sct1}} 

Quantum error correction is an important part of various schemes
for quantum
computation and quantum communication, and hence quantum error
correcting
codes, first introduced about a decade ago
\cite{PhysRevA.52.R2493,PhysRevA.55.900,PhysRevLett.77.793} have
received a great deal of attention. For a detailed discussion
see Ch.~10 of
\cite{NielsenChuang:QuantumComputation}. Most of the early work
dealt with codes for qubits, with a
Hilbert space of dimension $D=2$, but qudit codes with $D>2$
have also been
studied \cite{IEEE.45.1827, IEEE.47.3065, PhysRevA.65.012308,
quantph.0111080, quantph.0202007,
IJQI.2.55,quantph.0210097,quantph.0211014,PhysRevA.78.042303,
PhysRevA.78.012306,PhysRevA.81.032326}.
They are of intrinsic interest and could turn out to be of some
practical value.

The stabilizer formalism introduced by Gottesman in
\cite{quantph.9705052} for $D=2$
(qubits) provides a compact and powerful way of generating
quantum error
correcting codes and extends the notion of linear classical
error correcting codes \cite{MacWilliamsSloane:TheoryECC} to the
quantum domain. The stabilizer formalism has been generalized to
cases where $D$ is prime or a prime
power, see e.g. \cite{IEEE.47.3065,
IEEE.51.4892,quantph.0211014,PhysRevA.78.062315}. For composite
$D$ things are more complicated and there is no immediate and
natural way of generalizing the notions. Our approach is to use
generalized Pauli operators and stabilizers defined in the same
way as in the prime case, see e.g.
\cite{PhysRevA.78.042303,PhysRevA.81.032326}. This has the
virtue that many (although not all) results that are valid in
the prime dimensional case can be extended without too much
difficulty to the more general composite case.

An important problem in the theory of stabilizer codes is what
is their structure. Is there any ``canonical" way of
representing an arbitrary stabilizer code? If yes, can one use
this fact for implementing various quantum error-correcting
tasks? For prime $D$ it turns out that there is such a standard
form, see e.g. Ch. 10.5.7 of
\cite{NielsenChuang:QuantumComputation}, and this allows for a
better understanding of the error-correcting capabilities of the
stabilizer code and also provides an efficient way of
constructing encoding/decoding circuits for such stabilizer
codes. For composite $D$ we are not aware of any such standard
form (except for the case of stabilizer codes over prime-power
finite fields \cite{quantph.0211014}), and the proof that such a
form exists is one of the main results of the current article.

The reminder of the paper is organized as follows.
Sec.~\ref{sct2} contains definitions of the generalized Pauli
group and some quantum gates used later in the paper. It also
defines rigorously qudit stabilizers and their corresponding
stabilized subspaces (or codes), together with an alternative
algebraic notation that we employ later. Sec.~\ref{sct3} contains
our main results: a ``size" theorem that relates the size of the
stabilizer group to the dimension of its stabilized subspace,
followed by a ``structure" theorem that shows that any qudit
stabilizer can be brought to a standard form through a series
of elementary quantum gates. Finally, Sec.~\ref{sct4} contains a
summary, conclusions, and some open questions.

\section{Preliminary remarks and definitions}\label{sct2}
\subsection{The generalized Pauli group on $n$ qudits}\label{sbsct2A}

We generalize Pauli operators to higher dimensional systems of
arbitrary
dimension $D$ in the following way. The $X$ and $Z$ operators
acting on a
single qudit are defined as
\begin{equation}
\label{eqn1}
Z=\sum_{j=0}^{D-1}\omega^j\dyad{j}{j},\quad
X=\sum_{j=0}^{D-1}\dyad{j}{j+1},
\end{equation}
and satisfy
\begin{equation}
\label{eqn2}
X^D=Z^D=I,\quad XZ=\omega ZX,\quad \omega = \mathrm{e}^{2 \pi
\ii /D},
\end{equation}
where \emph{the addition of integers is modulo $D$}, as will be assumed from 
now on. For a collection of $n$ qudits we use
subscripts to
identify the corresponding Pauli operators: thus $Z_i$ and $X_i$
operate on
the space of qudit $i$. The Hilbert space of a single qudit is
denoted by
$\HC$, and the Hilbert space of $n$ qudits by $\HC_n$,
respectively. Operators of the form
\begin{equation}
\label{eqn3}
\omega^{\lambda}X^{\vect{x}}Z^{\vect{z}} :=
\omega^{\lambda}X_1^{x_1}Z_1^{z_1}\otimes
X_2^{x_2}Z_2^{z_2}\otimes\cdots
\otimes X_n^{x_n}Z_n^{z_n}
\end{equation} 
will be referred to as \emph{Pauli products}, where $\lambda$ is
an integer
in $\ZZ_D$ and $\vect{x}$ and $\vect{z}$ are $n$-tuples in
$\ZZ_D^n$, the
additive group of $n$-tuple integers mod $D$. For a fixed $n$
the collection
of all possible Pauli products \eqref{eqn3} form a group under
operator
multiplication, the \emph{Pauli group} $\PC_n$. If $p$ is a
Pauli product,
then $p^D=I$ is the identity operator on $\HC_n$, and hence the
order of any
element of $\PC_n$ is either $D$ or else an integer that divides
$D$. While
$\PC_n$ is not Abelian, it has the property that two elements
\emph{commute up
  to a phase}: 
\begin{equation}\label{eqn4}
p_1p_2 = \omega^{\lambda_{12}} p_2p_1, 
\end{equation}
with $\lambda_{12}$ an
integer in $\ZZ_D$ that depends on $p_1$ and $p_2$.

\subsection{Generalization of qubit quantum gates to higher
dimensions}\label{sbsct2B}


In this subsection we define some one and two qudit gates
generalizing
various qubit gates. The qudit generalization of the Hadamard
gate is the
\emph{Fourier gate}
\begin{equation}\label{eqn5}
\mathrm{F}:=\frac{1}{\sqrt{D}}\sum_{j=0}^{D-1}\omega^{jk}\dyad{j}{k}.
\end{equation}
For an invertible integer $q\in\ZZ_D$ (i.e. integer for which
there exists $\bar q\in\ZZ_D$ such that $q \bar q \equiv 1 \bmod
D$), we define a
\emph{multiplicative gate}
\begin{equation}\label{eqn6}
 \mathrm{S}_q:=\sum_{j=0}^{D-1}\dyad{j}{jq},
\end{equation}
where $qj$ means multiplication mod $D$. The requirement that
$q$ be
invertible ensures that $\mathrm{S}_q$ is unitary; for a qubit
$\mathrm{S}_q$ is just the identity.


For two distinct qudits $a$ and $b$ we define the CNOT gate as
\begin{equation}
\label{eqn7}
\mathrm{CNOT}_{ab}:=\sum_{j=0}^{D-1}\dyad{j}{j}_a\otimes
X_b^j=\sum_{j,k=0}^{D-1}\dyad{j}{j}_a\otimes \dyad{k}{k+j}_b,
\end{equation}
the obvious generalization of the qubit Controlled-NOT, where
$a$ labels the control qudit and $b$ labels the target qudit.
Next the SWAP gate is defined as
\begin{equation}
\label{eqn8}
\mathrm{SWAP}_{ab}:=\sum_{j,k=0}^{D-1}\dyad{k}{j}_a\otimes
\dyad{j}{k}_b.
\end{equation}
It is easy to check that SWAP gate is hermitian and does indeed
swap
qudits $a$ and $b$. Unlike the qubit case, the qudit SWAP gate
is not a
product of three CNOT gates, but can be expressed in terms of
CNOT gates and
Fourier gates as
\begin{equation}\label{eqn9}
\mathrm{SWAP}_{ab}=\mathrm{CNOT}_{ab}(\mathrm{CNOT}_{ba})^{\dagger}\mathrm{CNOT}_{ab}(\mathrm{F}_a^2\otimes
I_b),
\end{equation}
with 
\begin{equation}\label{eqn10}
(\mathrm{CNOT}_{ba})^{\dagger}=(\mathrm{CNOT}_{ba})^{D-1}=(I_a\otimes
\mathrm{F}_b^2)\mathrm{CNOT}_{ba} (I_a\otimes \mathrm{F}_b^2).
\end{equation}
Finally we define the generalized Controlled-phase or CP gate as\begin{equation}
\label{eqn11}
\mathrm{CP}_{ab}=\sum_{j=0}^{D-1}\dyad{j}{j}_a\otimes Z^j_b=
\sum_{j,k=0}^{D-1}\omega^{jk}\dyad{j}{j}_a\otimes\dyad{k}{k}_b.
\end{equation}
The CP and CNOT gates are related by a local Fourier gate,
similar to the qubit case
\begin{equation}\label{eqn12}
\mathrm{CNOT}_{ab}=(I_a\otimes \mathrm{F}_b) \mathrm{CP}_{ab}
(I_a\otimes \mathrm{F}_b)^\dagger,
\end{equation}
since $\mathrm{F}$ maps $Z$ into $X$ under conjugation (see
Table \ref{tbl1}).


The gates $\mathrm{F}$, $\mathrm{S}_q$, SWAP, CNOT and CP are
unitary
operators that map Pauli operators to Pauli operators under
conjugation, as
can be seen from Tables ~\ref{tbl1} and~\ref{tbl2}. They are
elements of the
so called \emph{Clifford group} on $n$ qudits
\cite{quantph.9802007,PhysRevA.71.042315}, the group of
$n$-qudit unitary
operators that leaves $\PC_n$ invariant under conjugation, i.e.
if $O$ is a
Clifford operator, then $\forall p\in\PC_n$,
$OpO^\dagger\in\PC_n$. From
Tables~\ref{tbl1} and~\ref{tbl2} one can easily deduce the
result of
conjugation by $\mathrm{F}$, $\mathrm{S}_q$, SWAP, CNOT and CP
on \emph{any}
Pauli product. 

\begin{table}
\begin{tabular}{|l|l|l|}
\hline
Pauli operator    & $\mathrm{S}_q$ & $\mathrm{F}$  \\
\hline
\hline
$Z$ & $Z^{q}$ & $X$\\
\hline
$X$ & $X^{\bar q}$ & $Z^{D-1}$\\
\hline
\end{tabular}
\caption{The conjugation of Pauli operators by one-qudit gates
$\mathrm{F}$ and $\mathrm{S}_q$ ($\bar q$ is the multiplicative
inverse of $q$ mod $D$).}
\label{tbl1}
\end{table}

\begin{table}
\begin{tabular}{|l|l|l|l|}
\hline
Pauli product & $\mathrm{CNOT}_{ab}$ & $\mathrm{SWAP}_{ab}$ &
$\mathrm{CP}_{ab}$\\
\hline
\hline
$I_a\otimes Z_b$ & $Z_a\otimes Z_b$ & $Z_a\otimes I_b$ &
$I_a\otimes Z_b$\\
\hline
$Z_a\otimes I_b$ & $Z_a\otimes I_b$ & $I_a\otimes Z_b$ &
$Z_a\otimes I_b$\\
\hline
$I_a\otimes X_b$ & $I_a\otimes X_b$ & $X_a\otimes I_b$ &
$Z_a^{D-1}\otimes X_b$\\
\hline
$X_a\otimes I_b$ & $X_a\otimes X_b^{D-1}$ & $I_a\otimes{X}_b$ &
$X_a\otimes Z_b^{D-1}$\\
\hline
\end{tabular}
\caption{
The conjugation of Pauli products on qudits $a$ and $b$ by
two-qudit
gates CNOT, SWAP and CP. For the CNOT gate, the first qudit $a$
is the
  control and the second qudit $b$ the target.}
\label{tbl2}
\end{table}

\subsection{Qudit stabilizer codes}\label{sbsct2C}

Relative to this group we define a \emph{stabilizer} code $\CC$
to be a $K\geq 1$-dimensional subspace of the Hilbert space
satisfying three conditions:

\begin{description}

\item[C1] There is a subgroup $\SC$ of $\PC_n$ such that
for \emph{every} $s$ in $\SC$ and \emph{every} $\ket{\psi}$ in
$\CC$
\begin{equation}
\label{eqn13}
 s \ket{\psi} = \ket{\psi}
\end{equation}

\item[C2] The subgroup $\SC$ is maximal in the sense that every
$s$ in $\PC_n$
for which \eqref{eqn13} is satisfied for all $\ket{\psi}\in\CC$
belongs to
  $\SC$.

\item[C3] The coding space $\CC$ is maximal in the sense that
any ket
$\ket{\psi}$ that satisfies \eqref{eqn13} for every $s\in\SC$
lies in
  $\CC$.
\end{description}

If these conditions are fulfilled we call $\SC$ the
\emph{stabilizer} of the
code $\CC$. That it is Abelian follows from the commutation
relation \eqref{eqn4}, since for $K>0$
there is some nonzero $\ket{\psi}$ satisfying \eqref{eqn13}.

Note that one
can always find a subgroup $\SC$ of $\PC_n$ satisfying C1 and C2
for any
subspace $\CC$ of the Hilbert space, but it might consist of
nothing but the
identity. Thus it is condition C3 that distinguishes stabilizer
codes from
nonadditive codes. A stabilizer code is uniquely determined by
$\SC$ as
well as by $\CC$, since $\SC$ determines $\CC$ through C3, so in
a sense the code and its stabilizer are dual
to each other.

\subsection{Stabilizer generators and equivalent algebraic
descriptions of qudit stabilizer codes}\label{sbsct2D}

Any stabilizer group can be compactly described using a set of
\emph{group generators}. A generator corresponds to a specific
Pauli product and can be completely specified, see \eqref{eqn3},
by a phase $\lambda$ and two $n$-tuples in $\ZZ_D^n$, $\vect{x}$
and $\vect{z}$. A collection of $k$ generators can therefore be
represented by a $k$-component \textit{phase vector} over
$\ZZ_D$ (that contains all $k$ phases) and a $k\times 2n$
\emph{parity-check matrix} over $\ZZ_D$ with rows corresponding
to the stabilizer generators. For example, the stabilizer
\begin{equation}\label{eqn14}
\SC=\langle\omega^2X_1^3Z_2^2,X_2^2\rangle
\end{equation}
corresponds to the phase vector $(2,0)$ and parity-check matrix
\begin{equation}\label{eqn15}
\vect{S}=\begin{pmatrix}
3 & 0 & \vline & 0 & 2 \\
0 & 2 & \vline & 0 & 0
\end{pmatrix}.
\end{equation}
The angular brackets in \eqref {eqn14} means ``group generated
by", i.e. the group obtained by all possible products of the
group generators.
We call the left $k\times n$ block of the parity-check matrix
the $X$-block, and the right $k\times n$ the $Z$-block, since
they describe the $X$ and $Z$ parts of the stabilizer
generators, respectively.

Note that if $D$ is a prime number, any stabilizer group can be
described using no more than $n$ generators. However, in
composite dimensions one can have more that $n$ generators but
no more than $2n$. For example, in $D=4$, the $n=1$ qudit
stabilizer $\SC=\langle X^2, Z^2 \rangle$ is generated by $2$
(and not 1) elements and specify the stabilizer state
$(\ket{0}+\ket{2})/\sqrt{2}$. There is no way of representing
this state using only 1 generator; $Z^2$ by itself stabilizes
both $\ket{0}$ and $\ket{2}$, hence everything in their span, so
the condition C3 is not satisfied, i.e. the coding space is not
maximal. The same kind of analysis holds for $X^2$. A more
rigorous analysis can be done using the Theorem~\ref{thm1} of
Sec.~\ref{sct3}, which implies for this example that the size of
the stabilizer group must be equal to 4, hence it must be
generated by a single generator of order 4 or two generators
each of order 2. By inspection it is easy to rule out the first
case, so indeed one \emph{must} use 2 generators.

A conjugation of a stabilizer group by a Clifford operation will
change the stabilizer group to an isomorphic group. This will
correspond to a \emph{column operation} on the parity-check
matrix of the stabilizer, together with a transformation of the
phase vector. On the other hand, the generator description of a
stabilizer group is not unique: one can multiply a generator by
another one and still get the same group. This kind of operation
corresponds to a \emph{row operation} on the parity-check
matrix, again keeping in mind that in general the phase vector
will modify. From now on for the sake of simplicity we will
ignore the phase vector, although in real applications one has
to keep track of the phases.

The following represent what we call \emph{elementary
row/column} operations: a) interchanging of
rows/columns, b) multiplication of a row/column by an
\emph{invertible} integer, c)
addition of any multiple of a row/column to a \emph{different}
row/column. The column operations can be realized by
conjugations of the stabilizer by the Clifford operations in
Table~\ref{tbl3}, and the row operations just ensure that the
stabilizer group remains the same, i.e. the new set of
generators generate the same stabilizer group and not a smaller
one.

\begin{table}
\begin{tabular}{|l|l|l|l|}
\hline
Gate & $X$-part & $Z$-part \\
\hline
\hline
$\mathrm{SWAP}_{ab}$ & Interchange columns & Interchange columns
\\
& $a$ and $b$ & $a+n$ and $b+n$\\
\hline
$\mathrm{S}_{q,a}$ & Multiply column $a$ by& Multiply column
$a+n$ \\
& invertible integer $q^{-1}$ & by invertible integer $q$\\
\hline
$(\mathrm{CNOT}_{ab})^m$ & Substract $m$ times & Add $m$ times
column\\
& column $a$ from column $b$ & $b+n$ to column $a+n$\\
\hline
\end{tabular}
\caption{
Conjugation by the above quantum gates correspond to elementary
column operations on the $X$ and $Z$ parts of the parity-check
matrix of a stabilizer code. For the CNOT gate, the first qudit
$a$ is the control and the second qudit $b$ the target. The
integer exponent $m$ means CNOT applied $m$ times (or,
equivalently, the $m$-th power of the CNOT gate).
}
\label{tbl3}
\end{table}

\section{Size-Structure theorems}\label{sct3}

The following theorem generalizes to composite $D$ a well-known
result for prime $D$ that relates the size of the stabilizer
group to the dimension of its stabilizer subspace. Although the
composite $D$ result of our next theorem may have been known by
the community (see e.g. the claim near the end of Sec. 3.6 of
\cite{quantph.9705052}), we have not yet seen a proof of it.

\begin{theorem}[Size]\label{thm1}
Let $\CC$ be an $n$-qudit stabilizer code with stabilizer $\SC$.
Then
\begin{equation}\label{eqn16}
K\times|\SC|=D^n,
\end{equation}
where $K$ is the dimension of $\CC$, $|\SC|$ is the size (or
order) of the stabilizer group $\SC$ and $D$ is the dimension of
the Hilbert space of one carrier qudit.
\end{theorem}
\begin{proof}
Define 
\begin{equation}\label{eqn17}
P:=\frac{1}{|\SC|}\sum_{s\in\SC}s. 
\end{equation}
We will first show that $P$ is the projector onto $\CC$.

It follows at once that 
\begin{equation}\label{eqn18}
P=P^\dagger=P^2,
\end{equation}
where the equalities follow from the group property of $\SC$, so
$P$ is an orthogonal projector. Let $\ket{\psi}$ be an arbitrary
vector that belongs to the stabilizer code $\CC$. Then
$s\ket{\psi}=\ket{\psi}$ $\forall s$ (see the condition C1
\eqref{eqn13} that a stabilizer code must satisfy), which
together with \eqref{eqn17} yields
\begin{equation}\label{eqn19}
P\ket{\psi}=\ket{\psi}.
\end{equation}
Therefore the subspace $\WC$ onto which $P$ projects includes
$\CC$, $\CC\subset \WC$.
Let us now choose an arbitrary $\ket{\phi}\in\WC$. Then 
\begin{equation}\label{eqn20}
\ket{\phi}=P\ket{\phi}=\frac{1}{|\SC|}\sum_{s\in\SC}s\ket{\phi}.\end{equation}
Multiply \eqref{eqn20} on the left by some arbitrary $t\in\SC$
and use the group property of $\SC$ to get
\begin{align}\label{eqn21}
t\ket{\phi}&=tP\ket{\phi}=\frac{1}{|\SC|}\sum_{s\in\SC}ts\ket{\phi}=\frac{1}{|\SC|}\sum_{s\in\SC}s\ket{\phi}\notag\\
&=P\ket{\phi}=\ket{\phi}.
\end{align}
Since $t$ was arbitrary we arrived at the conclusion that 
\begin{equation}\label{eqn22}
t\ket{\phi}=\ket{\phi}, \forall t\in\SC,
\end{equation}
which proves that $\ket{\phi}$ belongs to the stabilizer
subspace $\CC$, and this implies that $\WC\subset \CC$.
Hence $P$ projects strictly onto $\CC$ (and not on some larger
subspace that includes $\CC$). Its trace is just the dimension
$K$ of $\CC$,
\begin{equation}\label{eqn23}
\Tr(P)=K=\frac{1}{|\SC|}D^n.
\end{equation}
where we have used that fact that all Pauli products in
\eqref{eqn17} are traceless except the identity (that must
belong to the sum, since $\SC$ is a group), of trace $D^n$. This
concludes the proof.
\end{proof}
When $D$ is a composite integer the dimension of the stabilizer
code does not have to be a power of $D$ any more (as was the
case in the prime $D$ case), but can be any divisor of $D^n$. As
an illustrative example, consider the 1-qudit stabilizer
generated by $\SC=\langle Z^2\rangle$ in $D=4$. It is obvious
that $\SC$ stabilizes a $K=2$-dimensional code
$\CC$=span\{$\ket{0},\ket{2}$\}, and the size of the stabilizer
is $|\SC|=2$.

Whenever $D=2$ (qubits) it was shown in \cite{quantph.9705052}
(see also Ch. 10.5.7 of
\cite{NielsenChuang:QuantumComputation} for a detailed
discussion) that any stabilizer code can be put into a
``standard form" or ``canonical form", and this is very useful
for constructing encoding/decoding quantum circuits for
stabilizer codes. This result can be generalized at once to
prime $D$. However, for composite dimensions, it is not so
obvious how to do the generalization, and the main technical
difficulty is that $\ZZ_D$ is a ring (and not a field) and
therefore some integers are not invertible. However, in the
following Theorem we show that one can still apply a technique
similar to a Gaussian elimination over $\ZZ_D$ and put any
composite $D$ stabilizer code into a standard form similar to
the one of prime $D$ case.

\begin{theorem}[Standard form]\label{thm2}
Let $\CC$ be a $K$ dimensional $n$-qudit stabilizer code with
stabilizer $\SC$ generated by $k\leqslant 2n$ generators and
with corresponding parity-check matrix $\vect{S}$ of size
$k\times 2n$. Then $\SC$ is isomorphic through a conjugation by
a Clifford operation to another stabilizer $\SC'$, with parity
check matrix $\vect{S'}$ in \emph{standard form}

\begin{equation}\label{eqn24}
\vect{S'}=
\begin{array}{r} r\{ \\ k-r\{ \end{array} \!\!\!\!
\left( \begin{array}{cc|cc}
\raisebox{0ex}[1.5ex]{$\overbrace{\vect{M}}^r$} & 
\raisebox{0ex}[1.5ex]{$\overbrace{\vect{0}}^{n-r}$} & 
\raisebox{0ex}[1.5ex]{$\overbrace{\vect{Z}_1}^r$} & 
\raisebox{0ex}[1.5ex]{$\overbrace{\vect{Z}_3}^{n-r}$} \\
\vect{0} & \vect{0} & \vect{Z}_2 & \vect{Z}_4
\end{array} \right),
\end{equation}
where the dimensions of the block matrices are indicated by
curly brackets.

Here $\vect{M}=$ diag($m_1,\ldots, m_r$), with $1\leqslant
r\leqslant n$, is a diagonal matrix with all $m_j\neq 0$
divisors of $D$. The matrices $\vect{Z}_1$ and $\vect{Z}_2$
satisfy
\begin{align}
\vect{Z}_1\vect{M}&=\vect{M}\vect{Z}_1^T&\text{mod
}D\label{eqn25},\\
\vect{Z}_2\vect{M}&=\vect{0}&\text{mod }D\label{eqn26},
\end{align}
where $T$ in the exponent denotes the transpose, 
and the matrix $\vect{Z}_4$ is a diagonal rectangular matrix,
with diagonal elements divisors of $D$. The notation $\vect{0}$
denotes the zero-block matrix.
\end{theorem}

\begin{proof}
The key ingredient of the proof is the Smith normal form:
through a sequence of elementary row/column operations mod $D$
(see the discussion at the end of Sec.~\ref{sct3}, a
matrix $M$ over $\ZZ_D$ can be converted to the Smith normal
form
\cite{Newman:IntegralMatrices,Storjohann96nearoptimal} (see also
Sec. IV.B of \cite{PhysRevA.81.032326} for an example)
\begin{equation}
  \vect{M'}= \vect{V}\cdot\vect{M}\cdot\vect{W},
\label{eqn27}
\end{equation}
where $\vect{V}$ and $\vect{W}$ are invertible (in the mod $D$
sense) square
matrices, and $\vect{M'}$ is a diagonal rectangular matrix, with
diagonal elements divisors of $D$. The matrix $\vect{V}$
represents the row operations and $\vect{W}$ the column
operations.

Note that in our case all necessary column operations can be
realized by the corresponding gates in Table~\ref{tbl3}, and,
more important, they \emph{do not mix} the $X$ and $Z$ parts of
the parity-check matrix. Therefore, without being concerned with
what happens to the $Z$ part of the parity-check matrix
$\vect{S}$, we can put its $X$ part in the Smith normal form
(again we stress that this can be done because the $Z$ part to
not interfere with the $X$ part), and arrive at a parity-check
matrix of the form \eqref{eqn24}. Next by another series of
Clifford gates acting only on the last $(n-r)$ qudits one can
further put the $\vect{Z}_4$ matrix in its Smith normal form,
without modifying the $X$ part of the parity-check matrix (which
is already in Smith normal form), since there are only zeros on
the last $n-r$ columns of the $X$ part; note that the row
operations are done on the last $k-r$ rows, and again do not
modify the $X$ part of the parity-check matrix. Since the
elementary row operations do not change the stabilizer group and
the elementary column operations correspond to Clifford gates,
see Table~\ref{tbl3}, our whole transformation from $\SC$ to
$\SC'$ is a conjugation by a Clifford operation.

Finally note that conjugation by Clifford operations do not
change the commutation relations. It is easy to deduce from
\eqref{eqn3} that two Pauli products described by
$(\vect{x}|\vect{z})$ and $(\vect{x'}|\vect{z'})$ commute if and
only if
\begin{equation}\label{eqn28}
\vect{x}\cdot\vect{z'}=\vect{z}\cdot\vect{x'}\quad\text{mod }D,
\end{equation}
where the dot represents the usual inner product of two vectors
in $\ZZ_D$, e.g. the sum of the products of individual
components.
Using \eqref{eqn28} we observe at once that the final set of
generators commute if and only if \eqref{eqn25} and
\eqref{eqn26} hold, and this concludes the proof.
\end{proof}

It is proved in \cite{Storjohann96nearoptimal} that a $M\times
N$ matrix can be reduced to the Smith form in only
$\OC(M^{\theta-1}N)$ operations from $\ZZ_D$, where $\theta$ is
the exponent for matrix multiplication over the ring $\ZZ_D$,
i.e. two $M\times M$ matrices over $\ZZ_D$ can be multiplied in
$\OC(M^{\theta})$ operations from $\ZZ_D$. Using standard matrix
multiplication $\theta=3$, but better algorithms
\cite{CoppersmithWinograd} allow for $\theta=2.38$. This ensures
that our procedure outlined above is computationally efficient.

Whenever $D$ is prime its only non-zero divisor
is $1$, hence $\vect{M}$ is just the identity matrix and $\vect{Z}_4$ has only 0's and 1's on the diagonal. From \eqref{eqn25} and \eqref{eqn26}, $\vect{Z}_1$ must be symmetric and $\vect{Z}_2$ must be the zero matrix, hence our
standard form reduces to the one known for prime $D$'s \cite{quantph.9705052}. \footnote{Except the
fact that we also perform column operations by Clifford
conjugations, which further simplifies the standard form of
\cite{quantph.9705052}.}

Also for prime $D$ it was proven in \cite{quantph.0211014} that
any stabilizer group $\SC$ is Clifford equivalent to another
stabilizer $\SC'$ generated only by $Z$'s, $\SC'=\langle
Z_1,Z_2,\ldots,Z_k\rangle$ \footnote{Their result is more
general and holds for any finite field, provided one redefines
the Pauli operators in \eqref{eqn1} accordingly.}. This result
is not true for composite $D$'s, and one easy to see
counterexample is the 1-qudit stabilizer $\SC=\langle
X^2,Z^2\rangle$ in $D=4$, mentioned before in
Sec.~\ref{sbsct2D}.

\section{Conclusions and open questions}\label{sct4}

We studied stabilizer codes with carrier qudits of composite
dimension $D$. We proved a size theorem that relates the size of
the stabilizer group to the dimension of its stabilized code.
Furthermore, we have shown that any stabilizer code can be put
in a standard (canonical) form through a series of Clifford
gates, and we provided a constructive algorithm. Our result
generalizes what was known in the prime $D$ case and may be
useful in constructing efficient encoding/decoding quantum
circuits, following the procedures outlined in
\cite{quantph.9705052}.

Our approach was based on the generalized Pauli group introduced
by \eqref{eqn1} and \eqref{eqn2}. However, for composite
dimensions, this is not the only way of introducing Pauli
operators. An alternative way is to split the dimension in its
prime-power factors which will then induce a natural splitting
of the carrier qudits in subsystems of prime-power dimensions.
In each of these subsystems then one can define Pauli operators
using finite fields (any finite field is isomorphic to a
prime-power canonical representation), as done e.g. in
\cite{quantph.0211014}. Although this is the scope of future
work, we think it may be useful since in a sense ``decouples"
the stabilizer into prime-power subsystems, and the latter can
be put into standard forms as done in \cite{quantph.0211014}.
One can then use previously known results for stabilizers over
finite fields to study various properties of composite $D$
stabilizers, and this may help building more efficient quantum
error correcting codes.

Finally one may ask if there exist alternative standard forms of qudit
stabilizer codes, perhaps more useful that the one presented
here. We do not know if such forms exist, and searching for them may be worthwhile.

\begin{acknowledgments} 
The research described here received support from the Office of
Naval Research and from the National Science Foundation through
Grant No. PHY-0757251.
\end{acknowledgments}


%

\end{document}